\documentclass[journal,twoside,web]{ieeecolor}
\usepackage{generic}
\usepackage{cite}
\usepackage{graphicx}
\usepackage{amsmath,amsfonts,mathrsfs}
\usepackage{hyperref}
\usepackage{float}
\usepackage{stfloats}
\usepackage[nocomma]{optidef}
\usepackage{booktabs}
\usepackage{makecell}
\usepackage{subfigure}
\usepackage[ruled,linesnumbered]{algorithm2e}
\usepackage{caption}
\usepackage{arydshln}
\allowdisplaybreaks[4]
\newtheorem{assumption}{Assumption}

\newtheorem{remark}{Remark}
\newtheorem{lemma}{Lemma}
\newtheorem{theorem}{Theorem}
\newtheorem{corollary}{Corollary}

\def\BibTeX{{\rm B\kern-.05em{\sc i\kern-.025em b}\kern-.08em
    T\kern-.1667em\lower.7ex\hbox{E}\kern-.125emX}}
\SetKwFor{Loop}{loop}{}{end\ loop}

\hypersetup{
colorlinks=true,
linkcolor=black
}

\begin{document}
\title{A Robust Data-Driven Iterative Control Method for Linear Systems with Bounded Disturbances}
\author{Kaijian Hu, Tao Liu %
\thanks{The work was supported by the Research Grants Council of the Hong Kong Special Administrative Region under the General Research Fund Through Project No. 17209219 and the National Natural Science Foundation of China through Project No. 62173287.}
\thanks{K. Hu and T. Liu are with the Department of Electrical and Electronic Engineering, The University of Hong Kong, Hong Kong SAR, and the HKU Shenzhen Institute of Research and Innovation, Shenzhen, China. (e-mail: kjhu@eee.hku.hk; taoliu@eee.hku.hk).}
}

\maketitle

\begin{abstract}
    This paper proposes a new robust data-driven control method for linear systems with bounded disturbances, where the system model and disturbances are unknown. Due to disturbances, accurately determining the true system becomes challenging using the collected dataset. Therefore, instead of designing controllers directly for the unknown true system, an available approach is to design controllers for all systems compatible with the dataset. To overcome the limitations of using a single dataset and benefit from collecting more data, multiple datasets are employed in this paper. Furthermore, a new iterative method is developed to address the challenges of using multiple datasets. Based on this method, this paper develops an offline and online robust data-driven iterative control method, respectively. Compared to the existing robust data-driven controller method, both proposed control methods iteratively utilize multiple datasets in the controller design process. This allows for the incorporation of numerous datasets, potentially reducing the conservativeness of the designed controller. Particularly, the online controller is iteratively designed by continuously incorporating online collected data into the historical data to construct new datasets. Lastly, the effectiveness of the proposed methods is demonstrated using a batch reactor.
\end{abstract}


\section{Introduction}
\IEEEPARstart{D}{irect} data-driven control (DDC) as an alternative to the traditional model-based control method is becoming popular in the control community\cite{VanWaarde2023}. This trend is driven by the difficulties in accurately modeling the system and the ease of collecting system data\cite{BerberichJulian2021LtMf}. Instead of relying on precise system models, DDC approaches focus on learning controllers directly from system data. In recent times, many DDC methods have been developed, including explicit feedback control \cite{DePersisClaudio2020FfDC, RotuloMonica2021Olod, FormentinSimone2016DloL, LuppiAlessandro2022Odso, Hu2023, Nortmann2023, vanWaardeHenkJ2020DIAN, vanWaardeHenkJ2020Fndt}, feedback linearized control\cite{Alsalti2023, DePersis2023L}, model-free adaptive control\cite{Hou2019}, model-free reinforcement learning control \cite{BUSONIU20188, Lopez2023}, data-driven optimal control\cite{Pan2023, Dörfler2023, CoulsonJeremy2018DPCI}, data-driven predictive control\cite{ BerberichJulian2021DMPC, Liu2023d, hu2024robust}, and iterative learning control\cite{Ahn2007}. These methods are applied to a wide range of systems, such as linear systems\cite{BerberichJulian2021DMPC, Liu2023d, CoulsonJeremy2018DPCI}, linear parameter-varying systems\cite{FormentinSimone2016DloL, Nortmann2023}, and switched linear systems\cite{RotuloMonica2021Olod, Hu2023}.

One of the critical challenges of the DDC methods is designing controllers for unknown systems subject to unknown disturbances\cite{vanWaardeHenkJ2020Fndt,hu2022}. This paper addresses this challenge in the context of linear time-invariant (LTI) systems. Due to disturbances, the pre-collected data of the system is difficult to determine its true model. Consequently, it is challenging to design a controller directly for the true system using the pre-collected data. Given this issue, a robust DDC method is developed using the matrix S-lemma in \cite{vanWaardeHenkJ2020Fndt}. This approach first constructs a system set encompassing all possible systems compatible with the pre-collected data. Subsequently, a controller is designed for each system within this set. Since the true system is included in the system set, the designed controller is also applicable to the true system. Based on this concept, several robust DDC controllers have been developed, such as $H_\infty$ control \cite{SteentjesTomRV2022Odci}, gain scheduling control\cite{Miller2023}, and tracking control\cite{Trentelman2022}.

Another crucial consideration when implementing DDC methods is the amount of data required. When dealing with LTI systems without disturbances, the pre-collected data is typically sufficient for DDC methods by satisfying the persistently exciting (PE) condition of a sufficiently high order \cite{DePersisClaudio2020FfDC, BerberichJulian2021DMPC, Hu2023}. However, this issue becomes more complex when disturbances are present in LTI systems. The existence of disturbances results in different datasets yielding different system sets, each potentially containing varying numbers of systems. Consequently, controllers designed based on these distinct system sets exhibit different control performances. Intuitively, using a larger amount of data can improve the control performance of the designed controller. However, methods utilizing a single dataset, such as \cite{vanWaardeHenkJ2020Fndt, SteentjesTomRV2022Odci}, may design a more conservative controller when more data is employed. Given this issue, multiple datasets are introduced in \cite{hu2022}. This paper constructs a minor system set by intersecting all system sets, each containing all systems compatible with a single dataset. Then, a robust DDC is designed only for all systems within the intersection set, which can reduce the conservativeness of the designed controller. Although this method does not answer how much data is required, it provides a method to make the designed controller achieve better control performance when more data is collected. 

To apply the method proposed in \cite{hu2022}, it is necessary to incorporate multiple datasets into a single linear matrix inequality (LMI), with each dataset requiring a decision variable. However, as the number of datasets increases, the size of the LMI also grows, increasing computational complexity. Additionally, if a new dataset is added, the LMI must be resolved with all previous and new datasets, further increasing the computational burden. Hence, the challenge of efficiently utilizing numerous datasets for the controller design remains unresolved.

On the other hand, most DDC methods primarily rely on pre-collected data to design controllers for unknown systems. However, in practical scenarios, large amounts of data can be collected online during the control process. Leveraging this data can potentially improve the designed controller's control performance. Recently, several DDC methods have emerged that incorporate online data, such as data-driven predictive control for LTI systems\cite{CoulsonJeremy2018DPCI, BerberichJulian2021DMPC, Liu2023d} and online learning control for switched linear systems\cite{RotuloMonica2021Olod}. The former preserves the pre-collected data that describes the system's behavior, while the latter solely relies on online data for describing the system's behavior. Moreover, these methods are primarily designed for systems without considering disturbances. Therefore, combining online and pre-collected data to design a controller for LTI systems with disturbances is an open issue.  

To address these issues, this paper first presents a new iterative method to solve the LMI with numerous datasets. This method allows for considering multiple datasets without introducing additional decision variables. Based on this iterative method, this paper aims to develop an offline and online data-driven $H_\infty$ control method for unknown LTI systems with bounded disturbances, respectively. The offline control method is designed for cases where only the controller designed at the end of the iteration is applied to the controlled system. In contrast, the online control method is designed for scenarios where the controller designed at each iteration is applied immediately. Notably, during the online control process, online data are continuously collected to construct new datasets. 

The remainder of this paper is organized as follows. Section \ref{sc-problemfor} describes the problem to be solved. Section \ref{sc-iterMethod} introduces a new iterative method. Based on this method, three different offline data-driven $H_\infty$ controllers are designed using the pre-collected data in Section \ref{sc-offlineControl}. Furthermore, Section \ref{sc-onlineControl} proposes an online data-driven $H_\infty$ controller using the online collected data. The effectiveness of the proposed methods is illustrated by a batch reactor in Section \ref{sc-caseStudy}. Finally, Section \ref{sc-conclusion} gives some concluding remarks.

\textbf{Notation:} Let $\mathbb{C}$, $\mathbb{R}$ $\mathbb{Z}$, and $\mathbb{N}$ represent the set of complex numbers, real numbers, integers, and natural numbers, respectively. Let $I_{n}$ denote the identity matrix of size $n\times n$, and $0_{n\times m}$ denote the zero matrix of size $n\times m$. Both subscripts can be omitted if they are clear from the context. For a matrix $M$, $M^\top$ denotes its transpose, $M^{-1}$ denotes its inverse when it is non-singular, $M^\dagger$ represents its pseudo-inverse, and $M<0$ ($M\leq 0$) indicates that it is negative (semi-)definite. For a collection of matrices $M_i$, $i=1,\dots,s$, $\text{diag}(M_1,\dots,M_s)$ represents the block diagonal matrix. For a signal ${z}(k):\mathbb{Z}\to \mathbb{R}^n$, define $\|z(k)\|_2$ as its $2$-norm, and ${z}_{[k,k+T]}=[{z}(k)^\top , {z}(k+1)^\top , \dots , {z}(k+T)^\top]^\top$, where $k\in\mathbb{Z}$ and $T\in\mathbb{N}$.


\section{Preliminaries and Problem Formulation}\label{sc-problemfor}
Consider a discrete-time multiple-input and multiple-output LTI system
\begin{subequations}\label{eq-system}
\begin{align}
x(k+1) &= A_t x(k)+B_t u(k)+E_t \omega(k),\label{eq-systema}\\
y(k)&=C_t x(k)+D_t u(k)+G_t \omega(k),\label{eq-systemb}
\end{align}
\end{subequations}
where $u(k)\in\mathbb{R}^m$ is the system input, $x(k)\in\mathbb{R}^n$ is the system state, $y(k)\in\mathbb{R}^p$ is the system output, and $\omega(k)=[\omega_p(k)^\top,\omega_m(k)^\top]^\top\in\mathbb{R}^r$ is system disturbance with $\omega_p(k)\in\mathbb{R}^n$, $\omega_m(k)\in\mathbb{R}^p$, and $r=n+p$. In addition, $A_t\in\mathbb{R}^{n\times n}$, $B_t\in\mathbb{R}^{n\times m}$, $C_t\in\mathbb{R}^{p\times n}$, $D_t\!\in\!\mathbb{R}^{p\times m}$, $E_t=[I_n,0_{n\times p}]\in\mathbb{R}^{n\times r}$, and $G_t=[0_{p\times n},I_p]\in\mathbb{R}^{p\times r}$ are the system matrices. 

Without loss of generality, we adopt the following two assumptions for the system (\ref{eq-system}), which are commonly used in the design of robust data-driven controllers\cite{vanWaardeHenkJ2020Fndt,hu2022}.
\begin{assumption}\label{as-as1}
    The dimensions of the system matrices $A_t$, $B_t$, $C_t$, and $D_t$ are known, i.e., $m$, $n$, and $p$ are known. However, the entries of these matrices are unknown.
\end{assumption}
\begin{assumption}\label{as-as3}
    The disturbance input $\omega(k)$ is bounded satisfying $\omega(k)\omega(k)^\top\leq\Upsilon$, where $\Upsilon\in\mathbb{R}^{r\times r}$ is a positive definite matrix.
\end{assumption}

Due to disturbances, multiple systems can generate the pre-collected data, making it challenging to directly design a controller for the true system based on the data. Instead, an alternative approach is to design a controller for all systems compatible with the pre-collected data. To accomplish this, a robust DDC method was recently designed in \cite{vanWaardeHenkJ2020Fndt} for LTI systems with bounded disturbances using the following matrix S-lemma.
\begin{lemma}[\cite{vanWaardeHenkJ2020Fndt}]\label{le-slemma}
    Let $M, N \in \mathbb{R}^{(s+l) \times(s+l)}$ be symmetric matrices. Suppose that there exists a matrix $\bar{Z}\in \mathbb{R}^{l \times s}$ such that $[I_l, \bar{Z}]N[I_l, \bar{Z}]^{\top}<0$. $[I_l, Z] M[I_l, Z]^{\top} \leq 0$ for all $Z \in \mathbb{R}^{l \times s}$ such that $[I_l, Z]$ $N [I_l, Z]^{\top} \leq 0$ if and only if there exists a scalar $\tau \geq 0$ such that 
    \begin{align}
        M-\tau N \leq 0.
    \end{align}
\end{lemma}

This method begins by constructing a set that contains all possible systems compatible with the pre-collected data. Then, a controller is designed for all systems in the set. Since the true system is included in the system set, the designed controller is also applicable to the true system. It is seen that the construction of the system set is pretty significant. Due to the disturbances, different datasets can construct different system sets. As a result, controllers designed based on these different system sets have different control performances. Intuitively, the more data is used, the better control performance of the designed controller can be achieved. However, methods relying on a single dataset, such as \cite{vanWaardeHenkJ2020Fndt, SteentjesTomRV2022Odci}, may design a more conservative controller when more data is utilized.

To clarify this problem, we present an example using the following simple system
\begin{align}
    \bar{x}(k+1)=a_t \bar{x}(k)+b_t \bar{u}(k)+\bar{\omega}(k),\nonumber
\end{align}
where $\bar{x}(k)\in\mathbb{R}$ is the state; $\bar{u}(k)\in\mathbb{R}$ is the input; $\bar{\omega}(k)\in\mathbb{R}$ is the disturbance, which is bounded in $[-0.5,0.5]$; $a_t\in\mathbb{R}$ and $b_t\in\mathbb{R}$ are unknown system parameters. We assume that $a_t=0.5$ and $b_t=1$, which are used to generate the following length $9$ input-state trajectory
\begin{small}
    \begin{align}
      \bar{u}_{[0,8]}&\!=\!\!\{0.62,0.81,-0.74,0.82,0.26,-0.80,-0.44,0.09,0.91\},\nonumber\\
      \bar{x}_{[0,9]}&\!=\!\!\{0.1,1.14,1.36,0.23,0.58,0.48,-0.15,-0.23,0.44,1.29\}.\nonumber
    \end{align}
\end{small} 

Then, we define the following matrices
\begin{subequations}
    \begin{align}
        \bar U_{0|T-1} &= [\bar u(0),\bar u(1),\dots,\bar u(T-1)],\label{eq-Uex}\\
        \bar X_{0|T-1} &= [\bar x(0),\bar x(1),\dots,\bar x(T-1)],\label{eq-Xex}\\
        \bar X_{1|T} &= [\bar x(1),\bar x(2),\dots,\bar x(T)],\label{eq-Xexp}\\
        \bar W_{0|T-1} &= [\bar \omega(0),\bar \omega(1),\dots,\bar \omega(T-1)],\label{eq-Wexp}
        \end{align}
\end{subequations}
where $T\in[1,9]$.

According to \cite{vanWaardeHenkJ2020Fndt}, the set containing possible systems compatible with (\ref{eq-Uex})-(\ref{eq-Xexp}) is defined as
  \begin{align}
    \Sigma_T=\{(a,b)|\bar W_{0|T-1} \bar W_{0|T-1}^\top\leq0.5^2T\},\nonumber
  \end{align}
where $a\!\in\!\mathbb{R}$ and $b\!\in\!\mathbb{R}$ are the possible values of $a_t$ and $b_t$, respectively, and $\bar W_{0|T-1}=\bar X_{1|T}-a\bar X_{0|T-1}-b\bar U_{0|T\!-1}$. By selecting a specific value for $T$ from $[1,9]$, the equation $\bar W_{0|T-1} \bar W_{0|T-1}^\top\leq0.5^2T$ represents an ellipse in the $ab$-plane. This ellipse has all possible $(a,b)$ values compatible with the corresponding data. It is evident that the size of the ellipse directly correlates with the number of possible systems contained within the set $\Sigma_T$. A larger ellipse indicates a larger number of possible systems within the set. In this case, we select $T=2, 3, 8$, and the corresponding ellipses are illustrated in Fig. \ref{fig-ellipse}. It can be observed that the ellipse size of $T=8$ is larger than that of $T=3$, which suggests that having more data does not guarantee a smaller set. As a result, utilizing more data may design a more conservative controller, which contradicts our intuition. Therefore, how to effectively utilize more data to design a less conservative controller
is an important research topic.

\begin{figure}[htbp]
    \centering
    \includegraphics[width=0.98\linewidth]{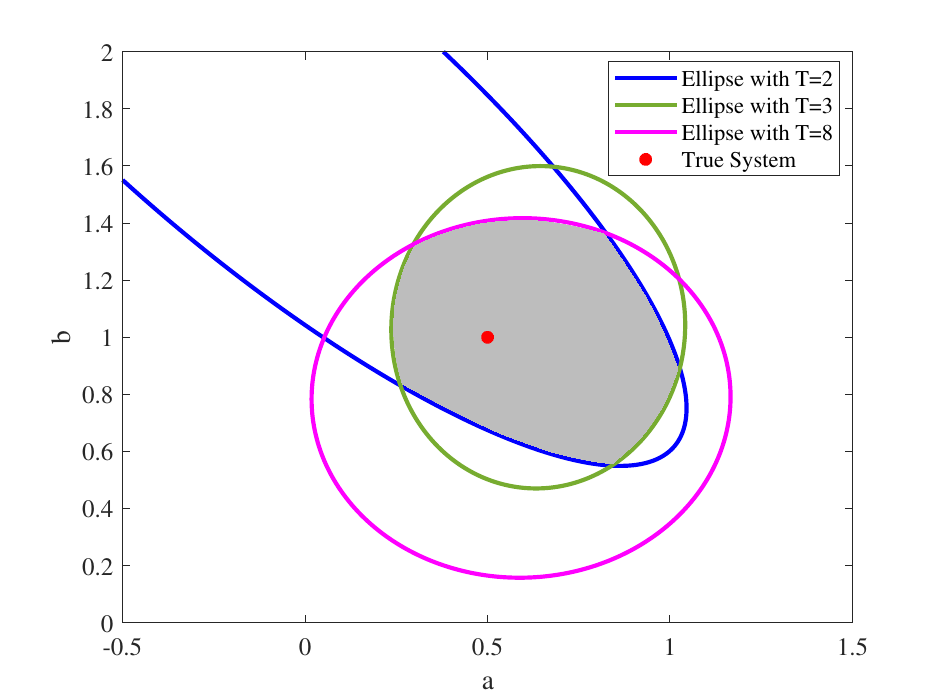}
    \caption{The ellipses with different data length $T$.}
    \label{fig-ellipse}
\end{figure}

To address this issue, we propose a method in \cite{hu2022} that designs the controller using multiple datasets instead of just one. Specifically, we first construct a set for each dataset and then design a controller only for systems in the intersection of these sets. Since the intersection set is smaller than any other set from a single dataset, exemplified in the last sub-figure of Fig. \ref{fig-ellipse}, the controller designed based on it could be less conservative. Additionally, the following Lemma \ref{le-slemma_v} are proposed to facilitate the controller design using multiple datasets.
\begin{lemma}[\cite{hu2022}]\label{le-slemma_v}
    Let $M,N_i \in \mathbb{R}^{(s+l)\times (s+l)}$, $i=1,\dots,q$ be symmetric matrices. If there exist constants $\tau_i\geq 0$, $i=1,\dots,q$ such that 
    \begin{align}
      M-\sum_{i=1}^{q}\tau_i N_i<0, \label{eq-multipleNi}
    \end{align}
    then $[I_l, Z] M [I_l, Z]^\top< 0$ holds for all $Z\in\mathbb{R}^{l\times s}$ such that 
    \begin{align}
        [I_l, Z]N_i [I_l, Z]^\top\leq 0, \forall i=1,\dots,q. \label{eq-le_qmi}
    \end{align}
\end{lemma}

In Lemma \ref{le-slemma_v}, the number of QMI \eqref{eq-le_qmi} and the unknown variables $\tau_i$ are the same as that of datasets. As the number of datasets increases, both the number of QMI \eqref{eq-le_qmi} and unknown variables $\tau_i$ also increase. Consequently, the scale of the LMI \eqref{eq-multipleNi} becomes larger, requiring more time to solve. This makes continuously increasing the number of datasets impractical, which hinders incorporating online data into the controller design process. Therefore, while this lemma allows for utilizing multiple datasets in controller design, it still has limitations when utilizing numerous datasets.

To overcome these limitations, we develop a new iterative method in this paper to solve the LMI \eqref{eq-multipleNi}. This approach allows for considering multiple datasets without introducing additional unknown variables. Based on this result, we design an offline and online data-driven $H_\infty$ control method for the system \eqref{eq-system}, respectively. The offline control method is designed for cases where only the controller designed at the end of the iteration is applied to the system \eqref{eq-system}. In contrast, the online control method is designed for scenarios where the designed controller at each iteration is applied immediately. Both control methods aim to achieve two control targets: i) stabilize the nominal system of (\ref{eq-system}); ii) attenuate the influence of the disturbance input $\omega(k)$ on the closed-loop system.


\section{A New Iterative Method}\label{sc-iterMethod}
As stated in Section \ref{sc-problemfor}, utilizing Lemma \ref{le-slemma_v} to design controllers for the system \eqref{eq-system} using numerous datasets presents a significant challenge. Furthermore, as more datasets are collected, the LMI (\ref{eq-multipleNi}) needs to be solved with increasing unknown variables $\tau_i$, $i=1,\dots,q$. To address this challenge, this section proposes a new iterative method for employing multiple datasets, as shown in Algorithm \ref{al-qmiIter}.
\begin{algorithm}[htbp]
    Set $i=1$, $\beta_1=0$, $N_1^h=0$\;
    \While{$i\leq q$}{
        Solve $M - \alpha_{i} N_{i}-\beta_i N_i^h< 0$\; \label{st-LMI}
        Set $N_{i+1}^h = \alpha_{i} N_{i}+\beta_{i} N_{i}^h$\;\label{st-hist-set}
        Set $i=i+1$\;
    }
    \caption{An Iterative Method to Solve LMI \eqref{eq-multipleNi}}\label{al-qmiIter}
\end{algorithm}

In Algorithm \ref{al-qmiIter}, $M$ and $N_i$ are defined in Lemma \ref{le-slemma_v}, $N_i^h\in\mathbb{R}^{(s+l)\times (s+l)}$, $\alpha_i\in\mathbb{R}$, and $\beta_i\in\mathbb{R}$, $i=1,\dots,q$. At each iteration $i$, only two datasets are considered: the new dataset $N_i$ and the historical dataset $N_{i}^h$. The historical dataset is constructed by combining the datasets that have been previously considered, as shown in step \ref{st-hist-set}.

Before implementing Algorithm \ref{al-qmiIter}, two crucial issues need to be addressed: i) whether the LMI in step \ref{st-LMI} is feasible at every iteration $i$, $1<i\leq q$; ii) whether a solution to the LMI \eqref{eq-multipleNi} can be obtained by implementing this algorithm. We summarize our findings in the following theorem.

\begin{theorem}\label{th-base}
Consider Algorithm \ref{al-qmiIter}. If there exists $\alpha_1\in\mathbb{R}$ such that the LMI in step \ref{st-LMI} is feasible at initial iteration $i=1$, then the following two statements hold.
\begin{enumerate}
    \item[(i)] The LMI is feasible at each iteration $i$, $1<i\leq q$;
    \item[(ii)] A solution to the LMI (\ref{eq-multipleNi}) is obtained at the end of the iteration.
\end{enumerate}
\end{theorem}
\begin{proof}
We first prove statement (i). Since the LMI in step \ref{st-LMI} at the initial iteration $i=1$, i.e., $M - \alpha_{1} N_{1}< 0$, is feasible, and $N_2^h=\alpha_{1} N_1$, the LMI in step \ref{st-LMI} is satisfied by selecting $\alpha_2=0$ and $\beta_2=1$ at the iteration $i=2$. This implies that the LMI is feasible at $i=2$. By repeating this process, we can conclude that the LMI is feasible at each iteration $i$, $1<i\leq q$ if it is feasible at $i=1$.

Subsequently, we prove the statement (ii). At the iteration $i=2$, substituting $N_2^h=\alpha_1 N_1$ into $M - \alpha_{2} N_{2}-\beta_2 N_2^h< 0$ yields
\begin{align}
    M - \alpha_{2} N_{2}-\alpha_1\beta_2 N_1< 0.
\end{align}

At $i=3$, we have $N_3^h=\alpha_2 N_2+\alpha_1\beta_2 N_1$. Substituting it into $M - \alpha_{3} N_{3}-\beta_3 N_3^h< 0$ yields
\begin{align}
    M - \alpha_{3} N_{3} - \alpha_{2} \beta_3 N_{2} - \alpha_1\beta_2\beta_3 N_1< 0.
\end{align}

Repeating the above process from $i=4$ to $i=q$ yields
\begin{align}
    M - \sum_{i}^q\alpha_i\prod_{j=i+1}^{q+1}\beta_j N_i<0,
\end{align}
where $\beta_{q+1}=1$.

It can be observed that $\tau_i=\alpha_i\prod_{j=i+1}^{q+1}\beta_j$, $i=1,\dots,q$ is a solution to the LMI (\ref{eq-multipleNi}).
\end{proof}

Subsequently, we utilize Algorithm \ref{al-qmiIter} to design data-driven $H_\infty$ control methods for the system \eqref{eq-system}. Depending on whether the controller designed at each iteration is applied to the controlled system immediately or not, we develop online and offline data-driven $H_\infty$ control methods, respectively.


\section{Offline Data-Driven $H_\infty$ Controller}\label{sc-offlineControl}
This section proposes an offline data-driven $H_\infty$ control method for the system \eqref{eq-system} using a group of pre-collected input-state-output datasets. To achieve this, we first construct a system set for each dataset, containing all possible systems compatible with the dataset. Then, we use Algorithm \ref{al-qmiIter} to design the controller by iteratively incorporating these system sets into the control design process.

Let $({u}_{i,[0,T_i-1]}, {x}_{i,[0,T_i]}, y_{i,[0,T_i-1]})$, $\forall i\in\mathcal{Q}=\{1,\dots,q\}$, represent the input-state-output trajectories of the system (\ref{eq-system}) collected under unknown disturbances $\omega_{i,[0,T_i-1]}$, where $q$ denotes the number of collected trajectories. Define the following data matrices
\begin{subequations}
    \begin{align}
        U_i &= [{u}_i(0), {u}_i(1), \dots, {u}_i(T_i-1)],\label{eq-Ui}\\
        X_{i} &= [{x}_i(0), {x}_i(1), \dots, {x}_i(T_i-1)],\label{eq-Xi}\\
        X_{i,+}&= [{x}_i(1), {x}_i(2), \dots, {x}_i(T_i)],\label{eq-Xip}\\
        Y_i &= [y_i(0), y_i(1), \dots, y_i(T_i-1)],\label{eq-Yi}\\
        W_i&= [\omega_i(0), \omega_i(1), \dots, \omega_i(T_i-1)].\label{eq-Wi}
\end{align}

\end{subequations}
\subsection{Offline Data-Based QMI}\label{sc-offlineQMI}
This subsection constructs a group of system sets $\Sigma_i$, $\forall i\in\mathcal{Q}$, where $\Sigma_i$ contains all possible systems compatible with (${u}_{i,[0,T_i-1]}$, ${x}_{i,[0,T_i]}$, $y_{i,[0,T_i-1]}$).

From (\ref{eq-system}), the data matrices (\ref{eq-Ui})-(\ref{eq-Wi}) satisfy
\begin{subequations}\nonumber
\begin{align}
X_{i,+}&=AX_{i}+B U_i+E_t W_i,\\
Y_i &= CX_{i}+DU_{i}+G_t W_i, \forall i\in\mathcal{Q},
\end{align}
\end{subequations}
which is equivalent to
\begin{align}
    \mathcal{G}W_i=[I_{n+p},Z]H_i, \forall i\in\mathcal{Q},\label{eq-datarel_sim}
\end{align}
where $H_i=[X_{i,+}^\top,Y_i^\top,-X_{i}^\top,-U_i^\top]^\top$ and 
\begin{align}
    Z&=
    \begin{bmatrix}
    A & B\\ C & D
    \end{bmatrix}, \mathcal{G}=\begin{bmatrix}
    E_t\\
    G_t
    \end{bmatrix}.\label{eq-ZGdef}
\end{align}
The notation $A\in\mathbb{R}^{n\times n}$, $B\in\mathbb{R}^{n\times m}$, $C\in\mathbb{R}^{p\times n}$, and $D\in\mathbb{R}^{p\times m}$ are the possible system matrices of $A_t$, $B_t$, $C_t$, and $D_t$, respectively. Notably, the datasets $(X_i, $ $X_{i,+}, Y_i, U_i)$, $\forall i\in\mathcal{Q}$ as well as matrices $E_t$ and $G_t$ are known. However, the disturbance sequence $W_i$, $\forall i\in\mathcal{Q}$ and the system matrices $(A, B, C, D)$ are unknown.

Under Assumption \ref{as-as3}, we have $\omega_i(k) {\omega_i(k)}^\top\leq \Upsilon$, which implies
\begin{align}
W_i{W_i}^\top\leq\Upsilon_{i}, \forall i\in\mathcal{Q},\label{eq-wwi}
\end{align}
where $\Upsilon_{i}=T_i\Upsilon$. 

Considering $i$-th set of the pre-collected data, for any $W_i$ satisfying \eqref{eq-wwi}, there may exist multiple system matrices $(A, B, C, D)$ satisfying (\ref{eq-datarel_sim}). All of these systems can be encompassed within the following set 
\begin{align}
    \Sigma_i=\{(A,B,C,D)|(\ref{eq-datarel_sim}) \text{ and } \eqref{eq-wwi}\}, \forall i\in\mathcal{Q}.\label{eq-set1}
\end{align}
It is evident that the true system matrices $(A_t,B_t,C_t,D_t)$ are also contained within ${\Sigma}_i$, $\forall i\in\mathcal{Q}$.

The set $\Sigma_i$ relies on the unknown $W_i$, making it difficult to use in practice. Therefore, we will eliminate $W_i$ in \eqref{eq-set1} by combining \eqref{eq-datarel_sim} and \eqref{eq-wwi}. Since the matrix $\mathcal{G}$ has full column rank, (\ref{eq-wwi}) is equivalent to \cite{Waarde2023}
\begin{align}
\mathcal{G}W_i{W_i}^\top \mathcal{G}^\top\leq \mathcal{G}\Upsilon_{i} \mathcal{G}^\top, \forall i\in\mathcal{Q}, \label{eq-wwehi}
\end{align}
which is obtained by pre- and post-multiplying both sides of (\ref{eq-wwi}) by $\mathcal{G}$ and $\mathcal{G}^\top$.

Substituting (\ref{eq-datarel_sim}) into (\ref{eq-wwehi}) gives the following quadratic matrix inequality (QMI)
\begin{align}
[I_{n+p}, Z] N_i [I_{n+p}, Z]^\top\leq 0, \forall i\in\mathcal{Q},\label{eq-condition_off}
\end{align}
where $N_i=H_i {H_i}^\top-\begin{bmatrix}
    \mathcal{G}\Upsilon_{i} \mathcal{G}^\top & 0\\
0 & 0
\end{bmatrix}$.

Then, we can redefine $\Sigma_i$ as
\begin{align}
\Sigma_i=\{(A, B, C, D)|(\ref{eq-condition_off})\}, \forall i\in\mathcal{Q}.\nonumber
\end{align}
The intersection of $\Sigma_i$, $\forall i\in\mathcal{Q}$ is defined as
\begin{align}
\Sigma^{off}=\bigcap_{i=1}^{q}\Sigma_i. \label{eq-sigma_off}
\end{align}
Since $(A_t,B_t,C_t,D_t)\!\in\!\Sigma_i,\forall i\in\mathcal{Q}$, $(A_t,B_t,C_t,D_t)\!\in\!\Sigma^{off}$. It implies that the controller designed for all systems in $\Sigma^{off}$ is applicable to the true system.

\subsection{Controller Design}
This subsection covers the design of offline data-driven $H_\infty$ controllers for all systems in $\Sigma^{off}$. Firstly, we design an offline data-driven $H_\infty$ controller by incorporating all datasets simultaneously. Based on this result, we then use Algorithm \ref{al-qmiIter} to design an offline data-driven $H_\infty$ controller by iteratively considering the datasets. Furthermore, we extend the control method to the scenario where the attenuation level is unknown. In this situation, we replace the LMI with an optimization problem to minimize the unknown attenuation level.

Let ${u}(k)=F{x}(k)$, where $F\in\mathbb{R}^{m\times n}$, be the state feedback controller of the system \eqref{eq-system}. The corresponding closed-loop system is given by
\begin{subequations}\label{eq-clstatespacei}
\begin{align}
{x}(k+1)&=\bar{A}{x}(k)+E_t\omega(k),\label{eq-clstatespaceia}\\
y(k)&=\bar{C}{x}(k)+G_t\omega(k),\label{eq-clstatespaceib}
\end{align}
\end{subequations}
where $\bar{A} = A+BF$ and $\bar{C} = C+DF$. With the help of Lemma \ref{le-slemma_v}, a state feedback controller can be designed by incorporating all datasets simultaneously, summarized in the following lemma.
\begin{lemma}\label{le-controller}
    Consider the system (\ref{eq-system}) under Assumptions \ref{as-as1}-\ref{as-as3}. Given a group of input-state-output trajectories $({u}_{i,[0,T_i-1]}, $ ${x}_{i,[0,T_i]}, y_{i,[0,T_i-1]})$, $\forall i\in\mathcal{Q}$ of \eqref{eq-system}, and an attenuation level $\gamma>0$, if there exist scalars $\tau_i\geq 0$, $\forall i\in\mathcal{Q}$, a matrix $S\in\mathbb{R}^{m\times n}$, and a positive definite matrix $\Gamma\in\mathbb{R}^{n\times n}$ such that the following LMI is feasible
    \begin{align}
        \mathcal{M}-\sum_{i=1}^{q}\tau_i \mathcal{N}_i< 0,\label{eq-dVwN}{}
    \end{align}
    where
    \begin{align}
    \mathcal{M}&\!=\!\!\begin{bmatrix}
    -\Gamma\!+\!\frac{1}{\gamma^2} E_t E_t^\top\! \!\!&\!\! \frac{1}{\gamma^2} E_t G_t^\top & 0 & 0 \!&\! 0\\
    \frac{1}{\gamma^2} G_t E_t^\top \!\!&\!\! -I_p\!+\!\frac{1}{\gamma^2} G_t G_t^\top & 0 & 0 \!&\! 0\\
    0 & 0 & \Gamma & S^\top \!&\! 0\\
    0 & 0 & S & 0 \!&\! S\\
    0 & 0 & 0 & S^\top \!&\! -\Gamma
    \end{bmatrix}\!\!,\label{eq-Mdef}\\
    \mathcal{N}_i&=\begin{bmatrix}
        N_i & 0\\
        0 & 0_{n\times n}
        \end{bmatrix},\label{eq-Nidef}
    \end{align}
    then the state feedback controller ${u}(k)=F x(k)$ with
    \begin{align}
    F=S\Gamma^{-1}\label{eq-Ksr}
    \end{align}
    asymptotically stabilizes the system (\ref{eq-system}) with $\omega(k)=0$. Furthermore, when the initial state is zero, i.e., $x(0)=0$, the designed controller achieves the following $H_\infty$ performance requirement
    \begin{align}
    \sum_{k=0}^{\infty}\|y(k)\|_{2}^2< \gamma^2 \sum_{k=0}^{\infty}\|\omega(k)\|_{2}^2.\label{eq-conditionH}
    \end{align}
\end{lemma}
\begin{proof}
    The proof is given in Appendix \ref{ap-pf_lemma3}.
\end{proof}

\begin{remark}
    The proposed controller has two advantages over the data-driven $H_\infty$ controller in \cite{vanWaardeHenkJ2020Fndt}. Firstly, the proposed controller is designed by incorporating multiple datasets, which may result in a less conservative controller. Secondly, the proposed controller investigates more general systems, where the matrices $C_t$ and $D_t$ are unknown, while \cite{vanWaardeHenkJ2020Fndt} assumes these matrices are known.
\end{remark}

In addition to designing the offline data-driven $H_\infty$ controller using all datasets simultaneously, as shown in Lemma \ref{le-controller}, an alternative approach is to iteratively design the controller by utilizing the datasets one by one. This iterative design process is similar to Algorithm \ref{al-qmiIter} and is outlined in Algorithm \ref{al-OfflineHinftyController}.
\begin{algorithm}[htbp]
    \KwData{Input-state-output datasets $({u}_{i,[0,T_i-1]}, {x}_{i,[0,T_i]},$  $y_{i,[0,T_i-1]}),\forall i\in\mathcal{Q}$}
    Given $\gamma$\;
    Set $i=1$, $\beta_1=0$, $\mathcal{N}_1^h=0$\;
    \While{$i\leq q$}{
        Construct $\mathcal{N}_{i}$ in \eqref{eq-Nidef}\;
        Solve the LMI \eqref{eq-off_pro}\;\label{st-lmi_off}
        Set $\mathcal{N}_{i+1}^h=\alpha_{i} \mathcal{N}_{i}+\beta_{i} \mathcal{N}_{i}^h$\;
        Set $i=i+1$\;
    }
    Set $F=S_q\Gamma_q^{-1}$\;
    \caption{Offline Data-Driven $H_{\infty}$ Controller}\label{al-OfflineHinftyController}
\end{algorithm}

In Algorithm \ref{al-OfflineHinftyController}, the LMI \eqref{eq-off_pro} is formulated as
\begin{align}
    \mathcal{M}_i-\alpha_i \mathcal{N}_i-\beta_i \mathcal{N}_i^h < 0,\label{eq-off_pro}
\end{align}
where $\alpha_i\geq 0$, $\beta_i\geq 0$, $\mathcal{N}_i$ is defined in \eqref{eq-Nidef}, and 
\begin{align}
    \mathcal{M}_i&=\text{diag}(Q_i,R_i),\nonumber\\
        \mathcal{N}_i^h&=\left\{
        \begin{array}{ccl}
            0, & & i=1,\\
            \alpha_{i-1} \mathcal{N}_{i-1}+\beta_{i-1} \mathcal{N}_{i-1}^h, & & 2\leq i\leq q.  
        \end{array}
            \right.\nonumber
\end{align}
In addition, the notations $Q_i$ and $R_i$ are defined as
\begin{align}
    Q_i &= \begin{bmatrix}
        -\Gamma_i\!+\!\frac{1}{\gamma^2} E_tE_t^\top & \frac{1}{\gamma^2} E_tG_t^\top\\
        \frac{1}{\gamma^2} G_tE_t^\top & -I_p+\frac{1}{\gamma^2} G_tG_t^\top   
    \end{bmatrix},\nonumber\\
    R_i & = \begin{bmatrix}
        \Gamma_i & S_i^\top & 0\\
        S_i & 0 & S_i\\
        0 & S_i^\top & -\Gamma_i    
    \end{bmatrix},\label{eq-Ridef}
\end{align}
where $\Gamma_i\in\mathbb{R}^{n\times n}$ is positive definite, and $S_i\in\mathbb{R}^{m\times n}$.

The following theorem summarizes the results regarding the stability of the closed-loop system (\ref{eq-system}) under the controller obtained through Algorithm \ref{al-OfflineHinftyController}.
\begin{theorem}\label{th-offhinftyControl}
Consider the system (\ref{eq-system}) under Assumptions \ref{as-as1}-\ref{as-as3}, and Algorithm \ref{al-OfflineHinftyController}. Given  a group of input-state-output trajectories $({u}_{i,[0,T_i-1]}, {x}_{i,[0,T_i]}, y_{i,[0,T_i-1]})$, $\forall i\in\mathcal{Q}$ of \eqref{eq-system}, and an attenuation level $\gamma\!>\!0$, if there exists a scalar $\alpha_1\!\geq\!0$, a matrix $S_1\!\in\!\mathbb{R}^{m\times n}$, and a positive definite matrix $\Gamma_1\in\mathbb{R}^{n\times n}$ at the initial iteration, i.e., $i=1$ such that the LMI \eqref{eq-off_pro} is feasible, then the LMI \eqref{eq-off_pro} is feasible at each iteration $i$, $1<i\leq q$. Furthermore, the state feedback controller ${u}(k)=F x(k)$ with $F=S_q\Gamma_q^{-1}$ asymptotically stabilizes the system (\ref{eq-system}) with $\omega(k)=0$, and meanwhile fulfills the $H_\infty$ performance requirement (\ref{eq-conditionH}) when the system has a zero initial state, i.e., $x(0)=0$.
\end{theorem}
\begin{proof}
Similar to the proof of Theorem \ref{th-base}, we can claim that the LMI (\ref{eq-off_pro}) is feasible at each iteration $1<i\leq q$ if it is feasible at $i=1$. Furthermore, at $i=q$, the pair $(\alpha_1\prod_{j=2}^q\beta_j,$ $\dots,\alpha_{q-1}\beta_q,\alpha_{q},S_q,\Gamma_q)$ is a solution for $(\tau_1,\dots,\tau_{q-1},\tau_q,S,$ $\Gamma)$ to the LMI (\ref{eq-dVwN}). According to Lemma \ref{le-controller}, the designed controller can achieve the two control targets.
\end{proof}

Theorem \ref{th-offhinftyControl} requires the attenuation level $\gamma$ to be specified. However, if the value of $\gamma$ is unknown beforehand, we can also design the controller using Algorithm \ref{al-OfflineHinftyController}. In this algorithm, $\gamma$ is treated as a variable in the LMI \eqref{eq-off_pro}. Moreover, it is possible to determine a minimum attenuation level $\gamma_i$ at each iteration $i$, $i\in\mathcal{Q}$ by replacing the LMI \eqref{eq-off_pro} in Algorithm \ref{al-OfflineHinftyController} with the following optimization problem \eqref{op-pr_off}.
\begin{mini!}|s|
    {\substack{\gamma_i,\alpha_i,\beta_i,S_i,\Gamma_i}}{\gamma_i^2}{\label{op-pr_off}}{}
    \addConstraint{\bar{\mathcal{M}}_i-\alpha_i \mathcal{N}_i-\beta_i \mathcal{N}_i^h < 0,}{\label{eq-off_pro_var}}
\end{mini!}
where $\gamma_i>0$, $\alpha_i\geq 0$, $\beta_i\geq 0$, $\bar{\mathcal{M}}_i=\text{diag}(\bar{Q}_i,R_i)$ with $R_i$ defined in \eqref{eq-Ridef} and 
\begin{align}
    \bar{Q}_i = \begin{bmatrix}
        -\Gamma_i+\frac{1}{\gamma_i^2} E_tE_t^\top & \frac{1}{\gamma_i^2} E_tG_t^\top\\
        \frac{1}{\gamma_i^2} G_tE_t^\top & -I_p+\frac{1}{\gamma_i^2} G_tG_t^\top   
    \end{bmatrix}.\nonumber
\end{align}

The results of the stability of the closed-loop system (\ref{eq-system}) under the designed controller and the attenuation level $\gamma_i$, $i\in\mathcal{Q}$, are summarized in the following corollary.
\begin{corollary}\label{co-offhinftyControl_gamma}
    Consider the system (\ref{eq-system}) under Assumptions \ref{as-as1}-\ref{as-as3}, and Algorithm \ref{al-OfflineHinftyController} that replaces the LMI \eqref{eq-off_pro} with the optimization problem \eqref{op-pr_off}. Given a group of input-state-output trajectories $({u}_{i,[0,T_i-1]}, {x}_{i,[0,T_i]}, y_{i,[0,T_i-1]})$, $\forall i\in\mathcal{Q}$ of \eqref{eq-system}, if there exists a scalar $\alpha_1\geq 0$, a matrix $S_1\in\mathbb{R}^{m\times n}$, and a positive definite matrix $\Gamma_1\in\mathbb{R}^{n\times n}$ at the initial iteration, i.e., $i=1$ such that the optimization problem \eqref{op-pr_off} is feasible, then the state feedback controller ${u}(k)=F x(k)$ with $F=S_q\Gamma_q^{-1}$ asymptotically stabilizes the system (\ref{eq-system}) with $\omega(k)=0$, and meanwhile fulfills the $H_\infty$ performance requirement \eqref{eq-conditionH} with $\gamma=\gamma_q$ when the system has a zero initial state. Moreover, the attenuation level $\gamma_i$, $i\in\mathcal{Q}$ decreases over the iteration, i.e., $\gamma_1\geq\gamma_2\geq\dots\geq \gamma_q$.
\end{corollary}
\begin{proof}
    Let $\lambda_i=1/\gamma_i^2$, $\forall i\in\mathcal{Q}$. The optimization problem \eqref{op-pr_off} can be equivalently written as
    \begin{mini!}|s|
        {\substack{\lambda_i,\alpha_i,\beta_i,S_i,\Gamma_i}}{1/\lambda_i}{\label{op-pr_off_var2}}{}
        \addConstraint{\tilde{\mathcal{M}}_i-\alpha_i \mathcal{N}_i-\beta_i \mathcal{N}_i^h \leq 0,}\label{eq-off_pro_var2}{}
    \end{mini!}
    where $\tilde{\mathcal{M}}_i=\text{diag}(\tilde{Q}_i,R_i)$ with 
    \begin{align}
        \tilde{Q}_i = \begin{bmatrix}
            -\Gamma_i+\lambda_i E_tE_t^\top & \lambda_i E_tG_t^\top\\
            \lambda_i G_tE_t^\top & -I_p+\lambda_i G_tG_t^\top   
        \end{bmatrix}.\nonumber
    \end{align}
    The convexity of the objective function $1/\lambda_i$ when $\lambda_i>0$, combined with the LMI \eqref{eq-off_pro_var2}, indicates that the optimization problem \eqref{op-pr_off_var2} is convex and possesses a global optimization solution.

    Similar to the proof of Theorem \ref{th-base}, the problem \eqref{op-pr_off_var2} is feasible at each iteration $1<i\leq q$ if it is feasible at $i=1$. Moreover, we can derive that $(\lambda_i, 0, 1, S_i, \Gamma_i)$ is a solution for $(\lambda_{i+1}, \alpha_{i+1}, \beta_{i+1}, S_{i+1}, \Gamma_{i+1})$ to the LMI \eqref{eq-off_pro_var} at time $i+1$, $1\leq i\leq q-1$. Combining this result with the fact that the optimization problem \eqref{op-pr_off_var2} is convex, we have $\frac{1}{\lambda_i}\geq \frac{1}{\lambda_{i+1}}$, which together with $\lambda_i=1/\gamma_i^2$ and $\gamma_i>0$ implies that $\gamma_i\geq \gamma_{i+1}$, $1\leq i\leq q-1$. It follows that $\gamma_1\geq\gamma_2\geq\dots\geq \gamma_q$.

    The rest of the proof is the same as that of Theorem \ref{th-offhinftyControl}, and thus is omitted.
\end{proof}
\begin{remark}
    In contrast to the controller design method in \cite{hu2022}, which employs all datasets simultaneously, the controller proposed in Corollary \ref{co-offhinftyControl_gamma} is designed iteratively by using one dataset at a time.
\end{remark}


\section{Online Data-Driven $H_\infty$ Controller}\label{sc-onlineControl}
This section presents an online data-driven $H_\infty$ control method for the system \eqref{eq-system}. In each iteration, this method applies the designed controller to the system \eqref{eq-system} and collects online data to construct new datasets.

Let $({u}_{[-L, k-1]}, {x}_{[-L, k]}, y_{[-L, k-1]})$ be an input-state-output trajectory of the system (\ref{eq-system}) under unknown disturbances $\omega_{[-L,k-1]}$ at time $k$, where $L$ is the data length. Define the following data matrices
\begin{subequations}
    \begin{align}
        U_{k-L|k-1} &= [{u}(k-L),\dots,{u}(k-1)],\label{eq-Uk}\\
        X_{k-L|k-1} &= [{x}(k-L),\dots,{x}(k-1)],\label{eq-Xk}\\
        X_{k-L+1|k} &= [{x}(k-L+1),\dots,{x}(k)],\label{eq-Xkp}\\
        Y_{k-L|k-1} &= [y(k-L),\dots,y(k-1)],\label{eq-Yk}\\
        W_{k-L|k-1} &= [\omega(k-L),\dots,\omega(k-1)].\label{eq-Wk}
    \end{align}
\end{subequations}
These data matrices will be updated in the next time step $k+1$.

\subsection{Online Data-Based QMI}\label{sc-dbQMI}
This subsection constructs a system set $\Sigma_k$ that contains all possible systems compatible with (${u}_{[-L, k-1]}$, ${x}_{[-L, k]}$, $y_{[-L, k-1]}$) at time $k$. The next iteration will construct the system set $\Sigma_{k+1}$ using the updated data matrices.

From (\ref{eq-system}), the data matrices (\ref{eq-Uk})-(\ref{eq-Wk}) satisfy 
\begin{align}
X_{k-L+1|k}&=AX_{k-L|k-1}+B U_{k-L|k-1}+E_t W_{k-L|k-1},\nonumber\\
Y_{k-L|k-1}&=CX_{k-L|k-1}+D U_{k-L|k-1}+G_t W_{k-L|k-1},\nonumber
\end{align}
which is equivalent to
\begin{align}
    \mathcal{G}W_{k-L|k-1}=[I_{n+p}, Z]H_k,\label{eq-ehwdes}
\end{align}
where both $Z$ and $\mathcal{G}$ are defined in (\ref{eq-ZGdef}), and 
\begin{align}
H_{k}=[X_{k-L+1|k}^\top,Y_{k-L|k-1}^\top,-X_{k-L|k-1}^\top,-U_{k-L|k-1}^\top]^\top.\nonumber
\end{align}

By repeating a similar process as in (\ref{eq-wwi}) to (\ref{eq-condition_off}), we obtain the following QMI
\begin{align}
[I_{n+p}, Z] N_{k} [I_{n+p}, Z]^\top\leq 0,\label{eq-condition_on}
\end{align}
where $N_{k}=H_{k} H_{k}^\top-\begin{bmatrix}
    \mathcal{G}\Upsilon_{L} \mathcal{G}^\top & 0\\
0 & 0
\end{bmatrix}$ with $\Upsilon_{L}=L\Upsilon$.

Then, define the system set $\Sigma_{k}$ as
\begin{align}
\Sigma_{k}=\{(A, B, C, D)|(\ref{eq-condition_on})\}.\nonumber
\end{align}
The intersection of $\Sigma_{\iota}$, $\forall \iota=0,\dots,k$ is defined as
\begin{align}
\Sigma_k^{on}=\bigcap_{\iota=0}^{k}\Sigma_{\iota}. \nonumber
\end{align}

Similar to the set $\Sigma^{off}$ in \eqref{eq-sigma_off}, $\Sigma_k^{on}$ also contain the true system. In addition, from the definition of $\Sigma_k^{on}$, we have
\begin{align}
    \Sigma_{k+1}^{on}=\Sigma_k^{on} \cap \Sigma_{k+1}, \nonumber
\end{align} 
which implies $\Sigma_{0}^{on}\supseteq \Sigma_{1}^{on}\supseteq \dots \supseteq \Sigma_k^{on} \supseteq \Sigma_{k+1}^{on} $. It means that the intersection set $\Sigma_k^{on}$ will gradually decrease over time. 
\subsection{Controller Design}
This subsection presents the online data-driven $H_\infty$ control method. At each time $k$, a state feedback controller ${u}(k)=F_k{x}(k)$, where $F_k\in\mathbb{R}^{m\times n}$, is designed for all systems in $\Sigma_k^{on}$. Since the set $\Sigma_k^{on}$ gradually decreases over time, the designed controller could progressively be less conservative.

The closed-loop system under the state feedback controller is described by
\begin{subequations}\label{eq-clstatespace}
\begin{align}
{x}(k+1)&=\bar{A}_{k}{x}(k)+E_t\omega(k),\label{eq-clstatespace1}\\
y(k)&=\bar{C}_{k}{x}(k)+G_t\omega(k),\label{eq-clstatespace2}
\end{align}
\end{subequations}
with $\bar{A}_{k} = A+BF_k$ and $\bar{C}_{k} = C+DF_k$. With the help of Algorithm \ref{al-qmiIter}, the online controller can be calculated using the following Algorithm \ref{al-OnlineHinftyController}.

\begin{algorithm}[htbp]
    \KwData{Input-state-output dataset $({u}_{[-L, -1]}, {x}_{[-L, 0]}, $ $y_{[-L, -1]})$}
    Given $\gamma$\;
    Set $k=0$, $\beta_0=0$, $\mathcal{N}_0^h=0$\;
    \Loop{}{
        Collect $x(k)$ and $y(k)$\; 
        Construct $\mathcal{N}_{k}=\text{diag}(N_k,0_{n\times n})$ with $N_k$ being defined right after \eqref{eq-condition_on}\;
        Solve the optimization problem (\ref{op-pr_on})\;\label{st-op_on}
        Apply $u(k)=F_k x(k)$ with $F_k=S_k\Gamma_k^{-1}$\;
        Set $\mathcal{N}_{k+1}^h=\alpha_{k} \mathcal{N}_{k}+\beta_{k} \mathcal{N}_{k}^h$\;
        Set $k=k+1$\;
    }
    \caption{Online Data-Driven $H_{\infty}$ Controller}\label{al-OnlineHinftyController}
\end{algorithm}

In Algorithm \ref{al-OnlineHinftyController}, the optimization problem \eqref{op-pr_on} is formulated as
\begin{mini!}|s|
    {\eta_k,\alpha_k,\beta_k,S_k,\Gamma_k}{\eta_k\label{op-ob}}{\label{op-pr_on}}{}
    \addConstraint{\begin{bmatrix}
        -\eta_k & x(k)^\top\\
        x(k) & -\Gamma_{k}
    \end{bmatrix} \leq 0,}\label{eq-bound}{}
    \addConstraint{\mathcal{M}_k-\alpha_k \mathcal{N}_k-\beta_k \mathcal{N}_k^h < 0,}\label{eq-dVwN_on}{}
\end{mini!}
where $\eta_k\geq0$, $\alpha_k\geq 0$ and $\beta_k\geq 0$, $S_k\in\mathbb{R}^{m\times n}$, $\Gamma_k\in\mathbb{R}^{n\times n}$ is a positive definite matrix, and
\begin{align}
    \mathcal{M}_k&=\text{diag}(Q_k,R_k),\nonumber\\
    \mathcal{N}_k^h&=\left\{
    \begin{array}{ccl}
        0, & & k=0,\\
        \alpha_{k-1} \mathcal{N}_{k-1}+\beta_{k-1} \mathcal{N}_{k-1}^h, & & k\geq 1,  
    \end{array}
        \right.\nonumber
\end{align}
with
\begin{align}
    Q_k &= \begin{bmatrix}
        -\Gamma_k\!+\!\frac{1}{\gamma^2} E_tE_t^\top & \frac{1}{\gamma^2} E_tG_t^\top\\
        \frac{1}{\gamma^2} G_tE_t^\top & -I_p+\frac{1}{\gamma^2} G_tG_t^\top   
    \end{bmatrix},\nonumber\\
    R_k & = \begin{bmatrix}
        \Gamma_k & S_k^\top & 0\\
        S_k & 0 & S_k\\
        0 & S_k^\top & -\Gamma_k    
    \end{bmatrix}.\nonumber
\end{align}

The following theorem summarizes the results regarding the stability of the closed-loop system (\ref{eq-system}) under the controller obtained through Algorithm \ref{al-OnlineHinftyController}.
\begin{theorem}\label{th-on_hinftyControl}
Consider the system (\ref{eq-system}) under Assumptions \ref{as-as1}-\ref{as-as3}, and Algorithm \ref{al-OnlineHinftyController}. Given an input-state-output trajectory $({u}_{[-L, -1]}, {x}_{[-L, 0]}, $ $y_{[-L, -1]})$ of \eqref{eq-system}, and an attenuation level $\gamma>0$, if there exist non-negative scalars $\eta_0$, $\alpha_0$, and $\beta_0$, a matrix $S_0\in\mathbb{R}^{m\times n}$, and a positive definite matrix $\Gamma_0\in\mathbb{R}^{n\times n}$ at the initial time step, i.e., $k=0$ such that the optimization problem (\ref{op-pr_on}) is feasible, then the problem \eqref{op-pr_on} is feasible at each time $k\in\mathbb{N}$. Furthermore, the state feedback controller ${u}(k)=F_k x(k)$ with $F_k=S_k\Gamma_k^{-1}$ asymptotically stabilizes the system (\ref{eq-system}) when $\omega(k)=0$, and meanwhile fulfills the $H_\infty$ performance requirement (\ref{eq-conditionH}) when the system has a zero initial state, i.e., $x(0)=0$.
\end{theorem}
\begin{proof}
    Following a similar proof to that of Theorem \ref{th-base}, we can establish that the problem \eqref{op-pr_on} remains feasible at each time $k\in\mathbb{N}$ if it is feasible at $k=0$. 
    
    Applying Schur complement\cite{1994Lmii} to (\ref{eq-bound}) gives
    \begin{align}
        x(k)^\top \Gamma_k^{-1} x(k)\leq \eta_k.\label{eq-bound_Pk}
    \end{align}
    
    It is straightforward to obtain that $(x(k+1)^\top\Gamma_k^{-1}x(k+1),$ $0, 1, S_k, \Gamma_k)$ is a solution for $(\eta_{k+1}, \alpha_{k+1}, \beta_{k+1}, S_{k+1}, \Gamma_{k+1})$ to the LMIs \eqref{eq-bound}-\eqref{eq-dVwN_on} at time $k+1$, $\forall k\in\mathbb{N}$.

    By solving the optimization problem (\ref{op-pr_on}) from time step $0$ to $k$, similar to the proof of Theorem \ref{th-offhinftyControl}, we can derive that the pair $(\alpha_0\prod_{j=1}^k\beta_j,$ $\dots,\alpha_{k-1}\beta_k,\alpha_{k},S_k,\Gamma_k)$ is a solution for $(\tau_0,\dots,\tau_{k-1},\tau_k,S,$ $\Gamma)$ with respect to the following LMI
    \begin{align}
        \mathcal{M}-\sum_{\iota=0}^{k}\tau_{\iota} \mathcal{N}_{\iota}< 0,\label{eq-dVwN_k}
    \end{align}
    where $\mathcal{M}$ is defined in \eqref{eq-Mdef}.

    According to the proof of Lemma \ref{le-controller}, we get
    \begin{align}
     &x(k+1)^\top P_k x(k+1)-x(k)^\top P_k x(k) \nonumber\\
     <&\gamma^2 \omega(k)^\top \omega(k)-y(k)^\top y(k),\label{eq-Gam}
    \end{align}
    where $P_k=\Gamma_k^{-1}$.
    
    The convexity of the objective function $\eta_k$, combined with the LMIs \eqref{eq-bound}-\eqref{eq-dVwN_on}, indicates that the optimization problem \eqref{op-pr_on} is convex and possesses a global optimization solution. Hence, according to \eqref{eq-bound_Pk}, $\eta_{k+1}$ is minimized at time $k+1$, which equals $x(k+1)^\top \Gamma_{k+1}^{-1} x(k+1)$. Since $x(k+1)^\top\Gamma_k^{-1}x(k+1)$ is a solution for $\eta_{k+1}$ to the LMIs \eqref{eq-bound}-\eqref{eq-dVwN_on} at time $k+1$, we have
    \begin{align}
        x(k+1)^\top \Gamma_{k+1}^{-1} x(k+1)\leq x(k+1)^\top\Gamma_k^{-1}x(k+1),\nonumber
    \end{align}
    which is equivalent to
    \begin{align}
    x(k+1)^\top P_{k+1}  x(k+1)\leq x(k+1)^\top P_k x(k+1).\label{eq-Gamx}
    \end{align}
    
    Then, combining (\ref{eq-Gam}) and (\ref{eq-Gamx}) gives 
    \begin{align}
     &x(k\!+\!1)^\top P_{k+1}  x(k\!+\!1)-x(k)^\top P_k x(k)\nonumber\\
     <&\gamma^2 \omega(k)^\top \omega(k)-y(k)^\top y(k).\label{eq-Gamx1}
    \end{align}
    
    Select the following Lyapunov function candidate 
    \begin{align}
        \bar{V}(x(k))=x(k)^\top P_k x(k).
    \end{align}
    
    When $\omega(k)=0$, the inequality (\ref{eq-Gamx1}) yields 
    \begin{align}
        \bar{V}(x(k+1))-\bar{V}(x(k))< -\begin{Vmatrix}
            y(k)
        \end{Vmatrix}_2^2.\label{eq-Gamx2x}
    \end{align} 
    Since $\bar{V}(x(k))>0$ for all $x(k)\neq 0$ and $\bar{V}(x(k+1))-\bar{V}(x(k))< 0$, the designed controller asymptotically stabilizes the system (\ref{eq-system}) when $\omega(k)=0$.
    
    Particularly, let the initial state $x(0) = 0$. The inequality (\ref{eq-conditionH}) can be derived by summing (\ref{eq-Gamx1}) from $k=0$ to $k=\infty$.
\end{proof}
\begin{remark}\label{re-opti_pro}
    Upon examining the proof of Theorem \ref{th-on_hinftyControl}, it can be observed that both the objective function \eqref{op-ob} and the LMI \eqref{eq-bound} are necessary to ensure the stability of the closed-loop system \eqref{eq-system} under the controller designed using Algorithm \ref{al-OnlineHinftyController}. According to the switched system theory\cite{Daniel03}, if we replace the optimization problem \eqref{op-pr_on} in Algorithm \ref{al-OnlineHinftyController} with the LMI \eqref{eq-dVwN_on}, it is possible for the closed-loop system to become unstable.
\end{remark}
\begin{remark}
    To implement Algorithm \ref{al-OnlineHinftyController}, we require a pre-collected dataset, i.e., $({u}_{[-L, -1]}, {x}_{[-L, 0]}, y_{[-L, -1]})$, and the initial state $x(0)$ to initiate the first iteration. The subsequent iteration is prepared by collecting online data to update both the data matrices \eqref{eq-Uk}-\eqref{eq-Yk} utilized in \eqref{eq-dVwN_on} and the state utilized in \eqref{eq-bound}.
\end{remark}
\begin{remark}
    In addition to constructing a new dataset using \eqref{eq-Uk}-\eqref{eq-Yk}, two other methods are provided as alternatives. Firstly, instead of fixing the data length $L$, one approach is to use datasets with different data lengths at each iteration. Secondly, instead of immediately using the collected data and updating one pair of data in the data matrices \eqref{eq-Uk}-\eqref{eq-Yk}, another approach is to store the collected data in a data buffer. Then, a sequence can be randomly selected from the buffer at each iteration to replace the corresponding data matrices. By considering these alternative methods, we can expedite the process of shrinking the intersection of system sets, i.e., $\Sigma_k^{on}$. This accelerated shrinking can potentially enhance the control performance of the designed controller.
\end{remark}
\begin{remark}
    Compared to the offline data-driven $H_\infty$ control method developed in Algorithm \ref{al-OfflineHinftyController}, the online data-driven $H_\infty$ control method developed in Algorithm \ref{al-OnlineHinftyController} exhibits two main differences. Firstly, the online method applies the designed controller to the system \eqref{eq-system} at each time step, whereas the offline method only applies the controller designed at the end of the iteration. Secondly, the online controller is designed by solving an optimization problem that ensures the stability of the closed-loop system, as discussed in Remark \ref{re-opti_pro}. In contrast, the offline controller only needs to solve an LMI problem.
\end{remark}

\section{Case Study}\label{sc-caseStudy}
This section demonstrates the effectiveness of the proposed methods by utilizing an unstable batch reactor\cite{WalshGC2001Sonc}. By setting the sampling time to $0.1$s \cite{Liu2023}, we obtain a discrete-time LTI model with the same form as the system \eqref{eq-system}. The system matrices are given by
\begin{align}
    A_t &= \begin{bmatrix}
        1.178  &  0.002  &  0.512  & -0.403\\
        -0.052  &  0.662  & -0.011  &  0.061\\
         0.076  &  0.335  &  0.561  &  0.382\\
        -0.001  &  0.335  &  0.089  &  0.849\\
     \end{bmatrix},E_t=[I_4,0],\nonumber\\
    B_t &= \!\begin{bmatrix}
        0.005  & -0.088\\
        0.467  &  0.001\\
        0.213  & -0.235\\
        0.213  & -0.016\\
    \end{bmatrix}\!, C_t=\!\!\begin{bmatrix}
        1  & 0 \\
        0  & 1 \\
        1  & 0 \\
       -1  & 0
    \end{bmatrix}^\top\!\!\!\!,D_t\!=\!0,G_t\!=\![0,I_2],\!\nonumber
\end{align}
where $x(k)\in\mathbb{R}^4$, $y(k)\in\mathbb{R}^2$, $u(k)\in\mathbb{R}^2$, and $\omega(k)\in\mathbb{R}^6$. The system matrices $A_t$, $B_t$, $C_t$, and $D_t$ are presumed to be unknown and only used to build the test system in the simulation. The disturbance $\omega(k)$ is bounded by $\omega(k)\omega(k)^\top\leq 0.0014I_6$. In addition, the initial state is set as $x_0=[0.51,$ $0.39,-0.30,-0.28]^\top$ in the simulation.

For comparison, we present the simulation results for the model-based $H_\infty$ controller\cite{Chang2014} and data-driven $H_\infty$ controllers designed using Lemma \ref{le-controller}, Algorithm \ref{al-OfflineHinftyController}, and Algorithm \ref{al-OnlineHinftyController}, respectively. In addition, we choose the attenuation level $\gamma=10$ for all controllers.

To implement data-driven controllers, we collect 100 length $T=8$ input-state-output trajectories of the batch reactor, denoted as ($u_{i,[0,7]}$, $x_{i,[0,8]}$, $y_{i,[0,7]}$), $i=1,\dots,100$, where inputs are randomly generated within the range of $[-0.1, 0.1]$. The first pair of trajectories, i.e., $i=1$, is used to design the offline controller using Lemma \ref{le-controller} with $q=1$. It also serves as the initial dataset for designing the offline data-driven $H_\infty$ controller using Algorithm \ref{al-OfflineHinftyController} and the online data-driven $H_\infty$ controller using Algorithm \ref{al-OnlineHinftyController}. Additionally, for the online controller, the data length of each dataset is set as $L=8$, which is the same as that of the two offline controllers.

The simulation results are illustrated in Fig. \ref{Fig.stateCom}. In this figure, the green and blue lines represent the outputs with the data-driven $H_\infty$ controllers (DDC) using Lemma \ref{le-controller} for $q=1$ and $q=100$, labeled as DDC with $q=1$ and DDC with $q=100$, respectively. The black lines are the outputs with the offline data-driven $H_\infty$ controller using Algorithm \ref{al-OfflineHinftyController}, labeled as Offline DDC. The red lines are the outputs with the online data-driven $H_\infty$ controller using Algorithm \ref{al-OnlineHinftyController}, labeled as Online DDC. Finally, the magenta dashed lines are the outputs generated by the model-based $H_\infty$ control (MBC) method\cite{Chang2014}.

Upon comparing the green lines with the blue lines, it is evident that the control performance of the DDC with multiple datasets is better than that of the DDC with a single dataset. Moreover, comparing the black lines with the blue lines reveals that the DDC designed by incrementally adding datasets achieves a control performance comparable to that of the DDC designed using all datasets at once. When comparing the blue lines with the magenta dashed lines, we can conclude that the DDC with multiple datasets can achieve control performance similar to that of the model-based controller despite the absence of explicit model information. 

Lastly, it is observed that the online DDC, represented by the red lines, initially demonstrates inferior control performance compared to the other control methods, except for the DDC with $q=1$. However, after a few time steps, the online DDC can achieve control performance comparable to the other methods. The reason is that the online DDC approach, where the controller at the first time step is designed using only one dataset, similar to the DDC with $q=1$. In subsequent time steps, new data is collected to construct new datasets, which are then utilized in the controller design process. In addition, despite both the online DDC at the first time step and the DDC with $q=1$ using a single dataset for controller design, the former exhibits improved control performance. This improvement can be attributed to the fact that the online controller is designed by solving an optimization problem rather than an LMI problem.

\begin{figure}[htbp]
    \centering
    \includegraphics[width=0.98\linewidth]{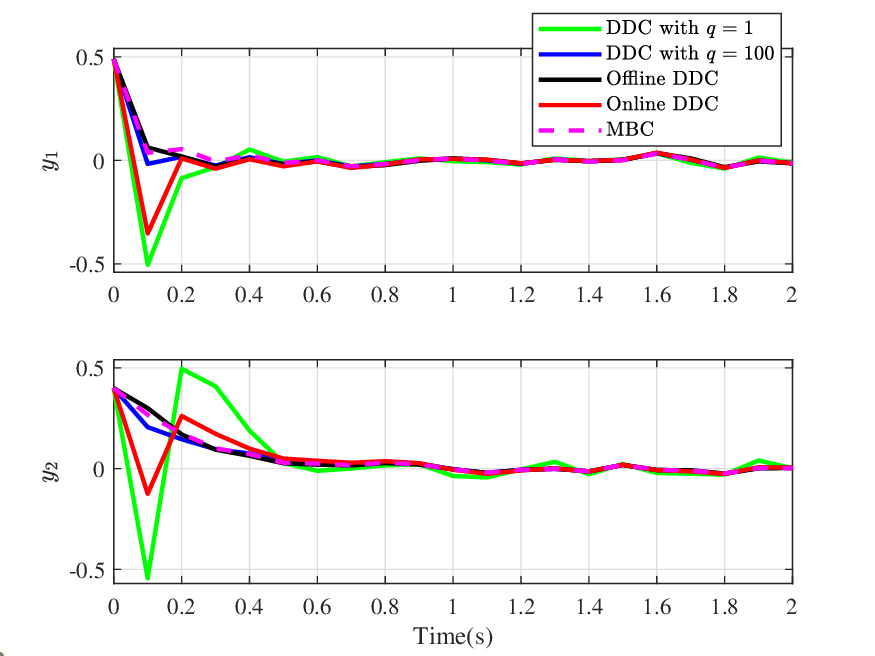}
    \caption{The outputs from DDC with $q=1$ and $q=100$, offline DDC, online DDC, and MBC, respectively.}
    \label{Fig.stateCom}
\end{figure}


\section{Conclusion}\label{sc-conclusion}
This paper has developed a new iterative method to solve the LMI with multiple datasets. Based on this method, two new data-driven $H_\infty$ control methods have been developed for LTI systems. The first is the offline control method, which only applies the controller designed at the end of the iteration to the controlled system. The second is the online control method, which immediately applies the designed controller to the controlled system at each iteration, and the online data are continuously collected to construct new datasets. Finally, an unstable batch reactor has been used to demonstrate the effectiveness of the proposed control methods.

\appendix
\subsection{Proof of Lemma \ref{le-controller}}\label{ap-pf_lemma3}
Select the following Lyapunov function candidate
\begin{align}
V({x}(k))={x}(k)^\top P {x}(k),\nonumber
\end{align}
where $P=\Gamma^{-1}$. Let $\Delta V({x}(k))=V(x(k+1))-V(x(k))$. Then, we have
\begin{align}
\Delta V({x}(k)) = {x}(k+1)^\top P {x}(k+1) - {x}(k)^\top P {x}(k).\label{eq-dV}
\end{align}
Substituting (\ref{eq-clstatespaceia}) into (\ref{eq-dV}) gives
\begin{align}
\Delta V({x}(k)) =\begin{bmatrix}
{x}(k)\\
\omega(k)
\end{bmatrix}^\top\begin{bmatrix}
\bar{A}^\top P \bar{A}-P & \bar{A}^\top P E_t\\
E_t^\top P \bar{A} & E_t^\top P E_t
\end{bmatrix}\begin{bmatrix}
{x}(k)\\
\omega(k)
\end{bmatrix}.\label{eq-dV_d}
\end{align}

Applying Schur complement\cite{1994Lmii} to (\ref{eq-dVwN}) gives
\begin{align}
M-\sum_{i=1}^{q}\tau_i N_i< 0, \label{eq-MtauN}
\end{align}
where $N_i$ is defined right after \eqref{eq-condition_off}, and
\begin{align}
    M=\begin{bmatrix}
        -\Gamma\!+\!\frac{1}{\gamma^2} E_t E_t^\top& \frac{1}{\gamma^2} E_t G_t^\top & 0 & 0\\
        \frac{1}{\gamma^2} G_t E_t^\top & -I_p+\frac{1}{\gamma^2} G_t G_t^\top & 0 & 0\\
          0 &  0 & \Gamma & S^\top\\
          0 &  0 & S      & S\Gamma^{-1} S^\top
        \end{bmatrix}.\nonumber
\end{align} 

Applying Lemma \ref{le-slemma_v} to (\ref{eq-MtauN}) yields that
\begin{align}
    [I_{n+p}, Z] M [I_{n+p}, Z]^\top< 0 \label{eq-Mcon}
\end{align}
holds for all $Z$ such that $[I_{n+p}, Z] N_i [I_{n+p}, Z]^\top \leq 0$, $\forall i=1,\dots,q$. It means that (\ref{eq-Mcon}) holds for all systems in $\Sigma^{off}$.

Substituting $Z$ in (\ref{eq-ZGdef}) and $F=S\Gamma^{-1}$, i.e., $S=F\Gamma$, into (\ref{eq-Mcon}) yields 
\begin{align}
    &\begin{bmatrix}
        -\Gamma\!+\!\frac{1}{\gamma^2} E_t E_t^\top& \frac{1}{\gamma^2} E_t G_t^\top\\
        \frac{1}{\gamma^2} G_t E_t^\top & -I_p+\frac{1}{\gamma^2} G_t G_t^\top 
    \end{bmatrix}\nonumber\\
    &+\begin{bmatrix}
        A & B\\
        C & D
    \end{bmatrix} \begin{bmatrix}
        \Gamma  & \Gamma F^\top\\
        F\Gamma & F\Gamma F^\top
    \end{bmatrix}\begin{bmatrix}
        A & B\\
        C & D
    \end{bmatrix}^\top<0,\nonumber
\end{align}
or equivalently,
\begin{align}
    &\begin{bmatrix}
        -\Gamma & 0\\
        0 & -I_{p}\\
    \end{bmatrix}+\begin{bmatrix}
        E_t \\
        G_t
    \end{bmatrix}
    \frac{1}{\gamma^2}I_{r}
    \begin{bmatrix}
        E_t \\
        G_t
    \end{bmatrix}^\top\nonumber\\
    &+\begin{bmatrix}
        A & B\\
        C & D
    \end{bmatrix} \begin{bmatrix}
        I_n\\
        F
    \end{bmatrix}
        \Gamma   
    \begin{bmatrix}
        I_n\\
        F
    \end{bmatrix}^\top  \begin{bmatrix}
        A & B\\
        C & D
    \end{bmatrix}^\top<0.\label{eq-Mbar}
\end{align}

From (\ref{eq-clstatespacei}), (\ref{eq-Mbar}) can be rewritten as
\begin{align}
    \!&\begin{bmatrix}
        -\Gamma & 0\\
        0 & -I_{p}\\
    \end{bmatrix}\!-\!\!\begin{bmatrix}
        \bar{A} \!&\! E_t \\
        \bar{C} \!&\!  G_t
    \end{bmatrix}\!\!
    \begin{bmatrix}
        -\Gamma \!&\! 0 \\
        0 \!&\! -\frac{1}{\gamma^2}I_{r}
    \end{bmatrix} \!\!
    \begin{bmatrix}
        \bar{A} \!&\! E_t \\
        \bar{C} \!&\!  G_t\
    \end{bmatrix}^\top\!\!\!<\! 0.\!\!\label{eq-phipsidxx}
\end{align}

Since $\text{diag}(-\Gamma,-\frac{1}{\gamma^2}I_{r})<0$, applying Schur complement to (\ref{eq-phipsidxx}) gives
\begin{align}
    \begin{bmatrix}
        \begin{array}{cc:ccc}
            -P & 0 & \bar{A}^\top & \bar{C}^\top \\
            0 & -\gamma^2 I_{r} & E_t^\top & G_t^\top\\
            \hdashline
            \bar{A} & E_t & -\Gamma & 0\\
            \bar{C} \!&\!  G_t & 0 & -I_{p}\\
        \end{array}
    \end{bmatrix} < 0. \label{eq-phipsidxx1}
\end{align} 

Since $\text{diag}(-\Gamma,-I_p)<0$, applying Schur complement to (\ref{eq-phipsidxx1}) gives
\begin{align}
    \begin{bmatrix}
    \bar{A}^\top P \bar{A}-P+\bar{C}^\top \bar{C} & \bar{A}^\top P E_t+\bar{C}^\top G_t\\
    E_t^\top P \bar{A}+G_t^\top \bar{C} & E_t^\top P E_t+G_t^\top G_t-\gamma^2 I_{r}
\end{bmatrix}\!\!<\!0,\nonumber
\end{align}
which is equivalent to 
\begin{align}
    \xi_k^\top\!\begin{bmatrix}
        \bar{A}^\top P \bar{A}-P+\bar{C}^\top \bar{C} & \bar{A}^\top P E_t+\bar{C}^\top G_t\\
    E_t^\top P \bar{A}+G_t^\top \bar{C} & E_t^\top P E_t+G_t^\top G_t-\gamma^2 I_{r}
\end{bmatrix}\xi_k\!<\!0\!\label{eq-phipsid}
\end{align}
for all non-zero $\xi_k=[{x}(k)^\top, \omega(k)^\top]^\top$. 

Combining (\ref{eq-dV_d}) with (\ref{eq-phipsid}) yields
\begin{align}
\Delta V(x(k))+\xi_k^\top\begin{bmatrix}
\bar{C}^\top \bar{C} & \bar{C}^\top G_t\\
G_t^\top \bar{C} & G_t^\top G_t-\gamma^2 I_{r}
\end{bmatrix}\xi_k< 0.\label{eq-deltaVp}
\end{align}
Substituting (\ref{eq-clstatespaceib}) into (\ref{eq-deltaVp}) gives
\begin{align}
\Delta V(x(k))< -\|y(k)\|_2^2 + \gamma^2 \|\omega(k)\|_2^2.\label{eq-deltaVx}
\end{align}

When $\omega(k)=0$, then $\Delta V(k) < -y(k)^\top y(k) \leq 0$. Thus, the proposed controller asymptotically stabilizes the nominal system of (\ref{eq-system}).

When $x(0)=0$, summing (\ref{eq-deltaVx}) from $k=0$ to $k=\infty$ derives (\ref{eq-conditionH}).

\section*{References}
\bibliographystyle{IEEEtran}
\bibliography{hukaijian}

\end{document}